\documentclass[11pt]{article}

\usepackage{graphics,graphicx,amssymb,amsmath,amsthm,latexsym,xspace,cases}
\usepackage{cleveref}[2012/02/15]
\usepackage{booktabs}
\usepackage{mathtools}
\usepackage[usenames]{color}
\usepackage[mathscr] {eucal}
\usepackage{todonotes}

\usepackage{latexsym}
\usepackage{amsfonts}

\usepackage{enumerate}

\newtheorem{theorem}{Theorem}[section]
\newtheorem{lemma}[theorem]{Lemma}

\theoremstyle{definition}

\newcommand{\F}{{\cal F}}

\newcommand{\Oh}{\mathrm{O}}

\newcommand{\ceil}[1]{\left\lceil{#1}\right\rceil}

\DeclarePairedDelimiter{\floor}{\lfloor}{\rfloor}

\newcommand{\dist}{\operatorname{dist}}
\newcommand{\distroot}{\operatorname{distroot}}

\newcommand{\nca}{\ensuremath{\mathrm{nca}}}

\usepackage{tikz}
\usepackage{tikz-qtree}
\usetikzlibrary{trees,matrix,fit,shapes,calc}

\newcommand{\set}[1]{\left \{ #1 \right \}}

\newcommand{\eps}{\varepsilon}

\newcommand{\scheme}{\sigma}

\usepackage{algorithm2e}

\title{Distance labeling schemes for trees}

\author{Stephen Alstrup   \thanks{Department of Computer Science,
University of Copenhagen. E-mail:
{\tt s.alstrup@di.ku.dk}.}
\and 
Inge Li G{\o}rtz \thanks{DTU Compute,
Technical University of Denmark. E-mail:
{\tt inge@dtu.dk}} 
\and
Esben Bistrup Halvorsen \thanks{ Department of Computer Science,
University of Copenhagen. E-mail: {\tt esben@bistruphalvorsen.dk}}. 
\and Ely Porat \thanks{Department of Computer Science,
	Bar-Ilan University.
{\tt porately@cs.biu.ac.il}}}
\begin{document}

\date{}
\maketitle

\pagestyle{empty}
\thispagestyle{empty}

\begin{abstract}
\noindent%
We consider distance labeling schemes for trees: given a tree with $n$ nodes, label the nodes with binary strings such that, given the labels of any two nodes, one can determine, by looking only at the labels, the distance in the tree between the two nodes.

A lower bound by Gavoille et. al. (J. Alg. 2004)  and an upper bound by Peleg (J. Graph Theory 2000) establish that labels must use $\Theta(\log^2 n)$ bits\footnote{Throughout this paper we use $\log$ for $\log_2$.}. Gavoille et. al. (ESA 2001) show that for very small approximate stretch, labels use $\Theta(\log n \log \log n)$ bits. Several other papers investigate various variants such as, for example, small distances in trees (Alstrup et. al., SODA'03).  

We improve the known upper and lower bounds of exact distance labeling by showing that $\frac{1}{4} \log^2 n$ bits are needed and that $\frac{1}{2} \log^2 n$ bits are sufficient. We also give ($1+\eps$)-stretch labeling schemes using $\Theta(\log n)$ bits for constant $\eps>0$. ($1+\eps$)-stretch labeling schemes with polylogarithmic label size have previously been established for doubling dimension graphs by Talwar (STOC 2004).

In addition, we present matching upper and lower bounds for distance labeling for caterpillars, showing that labels must have size $2\log n - \Theta(\log\log n)$. For simple paths with $k$ nodes and edge weights in $[1,n]$, we show that labels must have size $\frac{k-1}{k}\log n+\Theta(\log k)$.
\end{abstract}

\vspace*{3cm}

\newpage
\pagestyle{plain}
\setcounter{page}{1}


\section{Introduction}

A \emph{distance labeling scheme} for a given family of graphs assigns \emph{labels} to the nodes of each graph in the family such that, given the labels of two nodes in the graph and no other information, it is possible to determine the shortest distance between the two nodes. The labels are assumed to be composed of bits, and the goal is to make the worst-case label size as small as possible.  
 Labeling schemes are also called \emph{implicit representation of graphs}~\cite{spinrad2003efficient,wiki:implicit}. 
The problem of finding implicit representations with small labels for specific families of graphs was introduced  in the 1960s~\cite{Breuer66,BF67}, and efficient labeling schemes were introduced in~\cite{KNR92,muller}.
Distance labeling for general graphs has been considered since the 1970/80s~\cite{grahampollak,winkler}, and later for various restricted classes of graphs and/or approximate distances, often tightly related to distance oracle and routing problems, see e.g.~\cite{distalstrup}.
This paper focuses on distance labels for the well studied case of trees. 

\paragraph{Exact distances.}
In~\cite{peleg} Peleg presented an $\Oh(\log^2 n)$ bits distance labeling scheme for general unweighted trees. In~\cite{Gavoille200485} Gavoille et al.\ proved that distance labels for unweighted binary trees require $\frac{1}{8}\log^2 n -\Oh(\log n)$ bits and presented a scheme with $1/(\log 3-1) \log n \approx 1.7 \log n$ bits. This paper presents a  scheme of size $\frac{1}{2} \log^2 n +\Oh(\log n)$ and further reduces the gap by showing that $\frac{1}{4} \log^2 n - O(\log n)$ bits are needed. Our upper bound is a somewhat straightforward  application of a labeling scheme for nearest common ancestors~\cite{AGKR04,alstrupnca2014}. 

\paragraph{Approximate distances.}
Let $\dist_T(x,y)$ denote the shortest distance between nodes $x,y$ in a tree $T$. An $r$-additive approximation scheme returns a value $\dist'_T(x,y)$, where $\dist_T(x,y)\leq \dist'_T(x,y) \leq \dist_T(x,y)+r$. An $s$-stretched approximation scheme returns a value $\dist'_T(x,y)$, where $\dist_T(x,y)\leq \dist'_T(x,y) \leq \dist_T(x,y) \cdot s$. For trees of height $h$ Gavoille et al.~\cite[theorem 4]{GKKPP01} gave  a 1-additive $\Oh(\log n \log h)$ bit labeling scheme. However, using an extra bit in the label for the node depth modulo 2, it is easy to see that any 1-additive scheme can be made exact. Gavoille et al.~\cite{GKKPP01} also gave upper and lower bounds of $\Theta(\log \log n \log n)$ bits for $(1+1/\log n)$-stretched distance. This paper presents a scheme of size $\Theta(\log n)$ for $(1+\eps)$-stretch for constant $\eps>0$. Labeling schemes for $(1+\eps)$-stretch with polylogarithmic size label have previously been given for graphs of doubling dimension~\cite{Talwar04} and planar graphs~\cite{Thorup2004distance}. 

\paragraph{Distances in caterpillars and paths.}
Labeling schemes for caterpillars have been studied for various queries, e.g., adjacency~\cite{bonichon2006short}. Here we present upper and lower bounds showing that distance labeling caterpillars requires $2\log n -\Theta(\log \log n)$ bits. The upper bound is constructed by reduction to the case of \emph{weighted paths} with $k>1$ nodes and positive integer edge weights in $[1,n]$, for which we give upper and lower bounds showing that labels must have size $\frac{k-1}{k} \log n+\Theta(\log k)$.

\begin{table}[h!] 
	\centering
	\makebox[0pt][c]{
		\begin{tabular}{|c|c|c|}
			\hline
			\bf Problem & \bf Lower bound  & \bf Upper bound \\
			\hline
			Exact, general trees  & $\frac{1}{4}\log^2 n$  & $\frac{1}{2}\log^2 n$ \\
			\hline
			$(1+\eps)$-stretch, general trees  & \multicolumn{2}{c|}{$\Theta(\log n)$} \\
			\hline
			Caterpillars & \multicolumn{2}{c|}{$2\log n-\Theta(\log \log n)$} \\
			\hline
	         Weighted paths, $k$ nodes, weights in $[1,n]$ & \multicolumn{2}{c|}{$\frac{k-1}{k}\log n+\Theta(\log k)$} \\
			\hline
		\end{tabular}
	}
	\caption{Results presented in this paper. $\eps >0$ is a constant.
	}
	\label{tab:overview}
\end{table}

\subsection{Related work}

\paragraph{Distances in trees with small height.}
It is known that, for unweighted trees with bounded height $h$,  labels must have size $\Theta(\log n \log h)$. The upper bound follows from~\cite[Theorem 2]{GKKPP01} and the lower bound from~\cite[Section 3]{Gavoille200485}\footnote{We thank Gavoille for pointing this out.}. In~\cite{KaoLW07} distance labeling for various restricted class of trees, including trees with bounded height, is considered, and in~\cite{TangYZ09}  another distance labeling scheme for unweighted trees using $\Oh(\log n \log h)$ bits is given.

\paragraph{Small distances in trees.}
Distances in a tree between nodes at distance at most $k$ can be computed with labels of size $\log n + \Oh(k \sqrt{\log n} )$~\cite{KM01}. In~\cite{alstrupbillerauhe} it is shown that size $\log n +\Theta(\log \log n)$ are needed for labeling schemes supporting both parent and sibling queries. More generally, \cite{alstrupbillerauhe} shows that, using labels of size $\log n+\Oh(\log \log n)$, the distance between two nodes can be determined if it is at most $k$ for some constant $k$, which is optimal for $k>1$. In~\cite{GavoilleL07IS,GavoilleL07PO} further improvements are given for small distances in trees. For $k=1$, corresponding to adjacency testing, there is a sequence of papers that improve the second order term, recently ending with~\cite{adjacencytrees2015} which establishes that $\log n+\Theta(1)$ bits are sufficient.

\paragraph{Various other cases for trees.}
Distance labeling schemes for various other cases have been considered, e.g.,  for weighted trees~\cite{GKKPP01,Gavoille200485,peleg}, dynamic trees~\cite{dynamicKormanP07}, and a labeling scheme variation with extra free lookup~\cite{KormanK13,KormanK07}.

\paragraph{Exact and approximate distances in graphs.}
Distance labeling schemes for general graphs~\cite{distalstrup,Gavoille200485,grahampollak,spinrad2003efficient,weinmannpeleg,winkler} and various restricted graphs exist, e.g., for bounded tree-width, planar and bounded degree~\cite{Gavoille200485}, distance-hereditary~\cite{GP03b}, bounded clique-width~\cite{CV03}, some non-positively curved plane~\cite{CDV06}, interval~\cite{GP08} and permutation graphs~\cite{BG09}. Approximate distance labeling schemes, both additive and stretched,  are also well studied; see
e.g.,~\cite{CDEHV08,ElkinFN15,GKKPP01,GL05,Gavoille200485,GuptaKL03,gupta2005traveling,KL06,peleg,ThZw05}. An overview of distance labeling schemes can be found in~\cite{distalstrup}.

\subsection{Second order terms are important}\label{secondmatter}
Chung's solution in~\cite{Chung90} gives labels of size $\log n+O(\log \log n)$ for adjacency
labeling in trees, which was improved to $\log n + O(\log^* n)$  in FOCS'02~\cite{alstruprauhe} and in~\cite{bonichon2006short,Chung90,Fraigniaud2009randomized,fraigniaudkorman2,KMS02} to $\log n + \Theta(1)$ for various special cases. Finally it was improved to $\log n + \Theta(1)$ for general trees in FOCS'15~\cite{adjacencytrees2015}. 

A recent STOC'15 paper~\cite{AlstrupKTZ14} improves label size for adjacency in general graphs from $n/2+O(\log n)$~\cite{KNR92,moon1965minimal} to $n/2+ O(1)$ almost matching an $(n-1)/2$ lower bound~\cite{KNR92,moon1965minimal}.

Likewise, the second order term for ancestor relationship is improved in a sequence of STOC/SODA papers~\cite{AKM01,AlstrupBR03,AR02,fraigniaudkorman2,fraigniaudkorman} (and~\cite{abiteboul,KMS02})  to $\Theta(\log \log n)$, giving labels of size $\log n+\Theta(\log \log n)$.

Somewhat related, \emph{succinct data structures} (see, e.g.,~\cite{DPT10,FarzanM13,FarzanM14,MunroRRR12,patrascu08succinct})
focus on the space used in addition to the information theoretic lower bound, which is often a lower order term with respect to the overall space used.

\subsection{Labeling schemes in various settings and applications}
By using labeling schemes, it is possible to avoid costly access to large global tables, computing instead locally and distributed. Such properties are used, e.g., in XML search engines~\cite{AKM01}, network routing and distributed algorithms~\cite{Cowen01,EilamGP03,throupzwick,ThZw05}, dynamic and parallel settings ~\cite{CohenKaplan2010,dynamicKormanP07}, graph representations~\cite{KNR92}, and other applications~\cite{siamcompKatzKKP04,Korman2010,peleg2,peleg,SK85}.  Various computability requirements are sometimes imposed on  labeling schemes~\cite{AKM01,KNR92,siamcompKatzKKP04}. This paper assumes the RAM model.


\section{Preliminaries}

\paragraph{Trees.}
Given nodes $u,v$ in a rooted tree $T$, $u$ is an \emph{ancestor} of $v$ and $v$ is a \emph{descendant} of $u$, if $u$ is on the
unique path from $v$ to the root. For a node $u$ of $T$, denote by $T_u$ be the
subtree of $T$ consisting of all the descendants of $u$ (including
itself). The \emph{depth} of $u$ is the number of edges on the unique simple path from $u$ to the root of $T$. The nearest common ancestor (NCA) of two nodes is the unique common ancestor with largest depth. Let $T[u,v]$ denote the nodes on the simple path from $u$ to $v$ in $T$.
The variants $T(u,v]$ and $T[u,v)$ denote the same path without the first and last node, respectively.
 The distance between $u$ and $v$ is the number $\dist(u,v)=|T(u,v]|$. We set $\distroot(v)=\dist(v,r)$, where $r$ is the root of $T$. 
A \emph{caterpillar} is a tree whose non-leaf nodes form a path, called the \emph{spine}.

\paragraph{Heavy-light decomposition.} (From~\cite{sleatortarjan}.)
Let $T$ be a rooted tree. The nodes of $T$ are classified as either \emph{heavy} or \emph{light} as follows. The root $r$ of $T$ is light. For each non-leaf node $v$, pick one child  $w$ where $|T_w|$ is maximal among the children of $v$ and classify it as heavy; classify the other children of $v$ as light. 
The \emph{apex} of a node $v$ is the nearest light ancestor of $v$.
By removing the edges between light nodes and their parents, $T$ is divided into a collection of \emph{heavy paths}. Any given node $v$ has at most $\log n$ light ancestors (see~\cite{sleatortarjan}), so the path from the root to $v$ goes through at most $\log n$ heavy paths.

\paragraph{Bit strings.}
A bit string $s$ is a member of the set $\{0,1\}^*$. We denote the length of a
bit string $s$ by $|s|$, the $i$th bit of $s$ by $s_i$, and the concatenation of
two bit strings $s,s'$ by $s\circ s'$. We say that $s_1$ is the most significant
bit of $s$ and $s_{|s|}$ is the least significant bit.

\paragraph{Labeling schemes.}
An \emph{distance labeling scheme} for trees of size $n$ consists of an
\emph{encoder} $e$ and a \emph{decoder} $d$. Given a tree $T$, the
encoder computes a mapping $e_T : V(T)\to \set{0,1}^*$ assigning a \emph{label} to
each node $u\in V(T)$. The decoder is a mapping $d:\set{0,1}^*\times \set{0,1}^*\to \mathbb{Z}^+$, where $\mathbb{Z}^+$ denotes the positive integers, such that, given any tree $T$ and any
pair of nodes $u,v\in V(T)$, $d(e(u),e(v)) = \dist(u,v)$. Note that the decoder does not know $T$. The \emph{size} of a labeling scheme is defined as the maximum label size
$|e_T(u)|$ over all trees $T$ and all nodes $u\in V(T)$. If, for all trees $T$, the mapping $e_T$ is injective we say that the
labeling scheme assigns \emph{unique} labels.


\section{Distances on weighted paths}
In this section we study the case of  paths with $k$ nodes and integral edge weights in $[1,n]$. The solution to this problem will later be used to establish the upper bound for caterpillars.

\subsection{Upper Bound}

\begin{theorem} \label{theo:distanceweightedpathi}
	There exist a distance labeling scheme for paths with $k$ nodes and positive integral edge weights in $[1,n]$ with labels of size
  $\frac{k-1}{k}\log n +O(\log k)$.
\end{theorem}

\begin{proof}
We begin by considering the family of paths with $k$ nodes, integral edge weights and diameter $<n$. We shall prove that there exists a distance labeling scheme for this family with  labels of size $\frac{k-1}{k}\log n +\log k+\Oh(\log\log k)$.

	So consider such a path, and root it in one of its end nodes, denoted  $v_0$. Denote the nodes on the path $v_0,  \dots, v_{k-1}$ in order. Let $d_i=\distroot(v_i)$ and note that, by assumption,  $d_i<n$ for all $i$. We will let the label for $v_i$ store the number $d_i+x$ for some $x<n$ that allows us to represent $d_i+x$ compactly. Since we use the same $x$ for all nodes, we can easily compute the distance between any pair of nodes $v_i,v_j$ as $|(d_i+x)-(d_j+x)|$.

Since we choose $x<n$, the largest number stored in a label will be $d_k+x<2n$, which can be represented with \emph{exactly}  $L=\ceil{\log (2n)}$ bits.
	Divide those $L$ bits to $k+1$ segments, whereof $k$ have $\ell=\floor{L/k}$ bits and the last segment contains the remaining bits. The first segment, segment $0$, will contain the $\ell$ least significant bits, segment $1$ the following $\ell$ bits and so on.
We will choose $x$ such that the representation of $d_i+x$ has $0$s in all the bits in the $i$'th segment. If we manage to do so, we will be able to encode each $d_i+x$ with $L -\ell +\ceil{\log k}$ bits. Indeed,  we can use exactly $\ceil{\log k}$ bits to represent $i$, and the next $L-\ell$ bits to represent $d_i+x$ where we skip the $i$'th segment. Prefixing with a string in the form $0^{\ceil{\log\ceil{\log k}}}1$, we get a string from which we can determine the number of bits needed to write $\ceil{\log k}$ and therefrom the numbers $i$ and $d_i+x$. We use this string as the label for $v_i$. The label length is $L-\ell+\ceil{\log k}+\ceil{\log\ceil{\log k}}+1 = \frac{k-1}{k}\log n+\log k+O(\log\log k)$.
	
	It remains to show that there exist a number $x< n$ as described. In the following we shall, as in the above, represent numbers $<2n$ with $L$ bits that are divided into $k+1$ segments whereof the first $k$ have size $\ell$. For $i<k$ and $y<2n$, let $a(i,y)$ be a function which returns a number $z$ with the following properties:
	\begin{itemize}
		\item[(i)] In $z$, all bits outside segment $i$ are $0$. 
		\item[(ii)] $z+y$  has only $0$s in segment $i$.
	\end{itemize}
This function is constructed as follows. If $y$ only has $0$s in segment $i$, let $a(i,y)=0$. Otherwise take the representation of $y$, zero out all bits outside segment $i$, reverse the bits in segment $i$ and add $v$ to the resulting number, where $v$ has a 1 in the least significant bit of segment $i$ and 0s in all other positions. 

Note that from (i) it follows that adding $z$ to any number will not change bits in less significant positions than segment $i$.
We can now scan through the nodes $v_0,\dots v_{k-1}$, increasing $x$ by adding bits to $x$ in more and more significant positions (in non-overlapping segments), as follows: 

	\begin{itemize}
		\item 
		Set $x=0$.
		\item
		For $i=1\dots ,k-1$, set  $x=x+a(i,x+d_i)$.
	\end{itemize}
	
	After iteration $i$ we have that $x+d_i$ in segment $i$ only has 0s, and in the following iterations, $1$s are only added to $x$ in more significant bit positions, meaning that $d_i+x$ continues to have only $0$s in segment $i$. 
	Since the segments are non-overlapping, we end up with $x<n$.

For the more general family of paths with $k$ nodes and edge weights in $[1,n]$, we simply note that the diameter of any path in this family is at most $kn$. Using the above result thus immediately gives us a labeling scheme with labels of  size $\frac{k-1}{k}\log n + O(\log k)$.
\end{proof}


\subsection{Lower bound}

\begin{theorem} \label{theo:distanceweightedpath}
Labeling scheme for distances on weighted paths with $k$ nodes and edge weights $[1,n]$  require $\frac{k-1}{k}\log n +\Omega(\log k)$ bits.
\end{theorem}

\begin{proof}
Let $\F$ denote the family of paths with $k$ nodes and integral edge weights in $[1,n]$.  We can construct all the members of $\F$ by selecting $(k-1)$ different edge weights in the range $[1,n]$, skipping the paths which have already been constructed by the reverse sequence of edge weights. With this construction we will at most skip half of the paths, and hence $|\F|\geq \frac{1}{2}n^{k-1}$.  Let the worst-case label size of an optimal distance labeling scheme for such paths have length $L$. The number of different labels with length at most $L$ is $N=2^{L+1}-1$. We can uniquely represent each of the paths in $\F$ with the collection of their labels, and hence $|\F| \leq {N \choose k}$.  Thus, we have found that $\frac{1}{2}n^{k-1} \leq {N \choose k}$.  Since ${N\choose k}\leq (Ne/k)^k$, it follows that $\frac{k-1}{k} \log n \leq \log N -\log k + O(1)$ and hence that $L\geq \frac{k-1}{k}\log n + \log k - O(1) $ as desired.

\end{proof}
Combining \Cref{theo:distanceweightedpath} with \Cref{theo:distanceweightedpathi} we see that distance labels for paths of $k$ nodes with integral weights in $[1,n]$ must have length $\frac{k-1}{k}\log n +\Theta(\log k)$.


\section{Distances in caterpillars}

\subsection{Upper bound}
\begin{theorem}
	There exist a distance labeling scheme for caterpillars with worst case label size $2\log n -  \log\log n + O(\log\log\log n)$.  
\end{theorem}
\begin{proof}
We will start by giving a simple $2\log n$ bits scheme and then improve it. 
The simple solution assigns two numbers to each node. The  nodes on the spine save $\distroot$ and the number $0$. The nodes not on the spine save their parent's $\distroot$ and a number that is unique among their siblings. The second number is required to distinguish  siblings, and hence determine if the distance between two nodes is $0$ or $2$. The worst-case label size for this solution is $2\log n+O(1)$.

To improve the solution, we split up the nodes on the spine into two groups: (1) nodes with $> \frac{n}{k}$ leaves and  (2) nodes with $\leq \frac{n}{k}$ leaves, for some parameter $k$ to be chosen later. 
We add the root to the first group no matter what. 
Note that the first group can contain at most $k$ nodes.

As before, all nodes save two numbers: $\distroot$ and the number $0$ for spine nodes or a number to distinguish siblings. The idea is to reduce label size with $\log k$ bits by using fewer bits for the first number for nodes in the first group and for the second number for nodes in the second group.

The nodes in the first group form a path with at most $k$ nodes and edge weights in $[1,n]$ (where each weight corresponds to the distance between the nodes in the original graph).
The algorithm from \Cref{theo:distanceweightedpathi} will add a number $x$, which is less than the diameter, which again is less than $n$, to the numbers representing the root distances of the nodes. 
Using this technique, we can, as seen in the proof of \Cref{theo:distanceweightedpathi}, encode the (modified) $\distroot$s of the nodes in the first group with only $\frac{k-1}{k}\log n + \log k+O(\log\log k)$ bits. This gives labels of size $\frac{2k-1}{k}\log n+\log k+O(\log\log k)$ for non-spine nodes whose parents are in the first group.

We will also add $x$ to the $\distroot$s of nodes in the second group, but since $x<n$ this will not change the label size by more than a single bit.
For non-spine nodes whose parents are in the second group, we need at most $\log n -\log k +O(1)$ bits for the second number, giving a total label size of $2\log n -\log k +O(1)$. 

Finally, since the two numbers that form a label now have different lengths, we need an additional $O(\log\log k)$ bits to determine when one number ends and the next begins. Indeed, it wil be possible to split up labels into their components if we know the number of bits used to write $\ceil{\log k}$, and we represent this number with $O(\log\log k)$ bits.

Setting $k=\frac{\log n}{2\log\log n}$, we now see that our worst-case label size is the maximum of 
\[
2\log n - \log k + O(\log\log k) = 2\log n - \log\log n + O(\log\log\log n)
\]
and
\begin{align*}
\frac{2k-1}{k}\log n+\log k + O(\log\log k) & = 2\log n -2 \log\log n + \log\log n + O(\log\log\log n) \\
 &= 2\log n-\log\log n + O(\log\log\log n).
\end{align*}
This proves the theorem.
\end{proof}

\subsection{Lower bound}
We present a technique that counts tuples of labels that are known to be distinct and compares the result to the number of tuples one can obtain with labels of size $L$. The technique may have applications to distance labeling for other families of graphs.

\begin{theorem} \label{theo:distancelowercaterpillars}
	For any $n\geq 4$, any distance labeling scheme for the family of caterpillars with at most $n$ nodes has a worst-case label size of at least $2\floor{\log n} - \floor{\log\floor{\log n}} -4$.
\end{theorem}
\begin{proof}
	Set $k=\floor{\log n}$ and $m=2^k$. Let $(i_1,\dots ,i_k)$ be a sequence of $k$ numbers from the set $\{1,\dots ,m/2\}$ with the only requirement being that $i_1=1$. Now consider, for each such sequence, the caterpillar whose main path has length $m/2$ and where, for $t=1,\dots ,k$, the node in position $i_t$ has $\floor{m/2k}$ leaf children (not on the main path). We shall refer to these children as the \emph{$t$'th group}. Note that two disjoint groups of children may be children of the same node if $i_t=i_s$ for some $s,t$. Each of these caterpillar has $m/2+k\floor{m/2k}\leq m\leq n$ nodes.
	
	Suppose that $\scheme$ is a distance labeling scheme for the family of caterpillars, and consider one of the caterpillars defined above. Given distinct nodes $u,v$ not on the main path, their distance will be $\dist(u,v)=|i_s-i_t|+2$, where $i_s$ and $i_t$ are the positions on the main path of the parents of $u$ and $v$, respectively. In particular, if $s=1$, so that $i_s=1$, then $\dist(u,v)=i_t+1$.
	Thus, if $\scheme$ has been used to label the nodes of the caterpillar, the number $i_t$ for a child in the $t$'th group can be uniquely determined from its label together with the label of any of the children from the first group. It follows that any $k$-tuple of labels $(l_1,\dots ,l_k)$ where $l_t$ is a label of a child in the $t$'th group uniquely determines the sequence $(i_1,\dots ,i_k)$. In particular, $k$-tuples of labels from distinct caterpillars must be distinct. Of course, $k$-tuples of labels from the same caterpillar must also be distinct, since labels are unique in a distance labeling scheme.
	
	Now, there are $(m/2)^{k-1}$ choices for the sequence $(i_1,\dots ,i_k)$, and hence there are $(m/2)^{k-1}$ different caterpillars in this form. For each of these, there are $\floor{m/2k}^k$ different choices of $k$-tuples of labels. Altogether, we therefore have $(m/2)^{k-1}\floor{m/2k}^k$ distinct $k$-tuples of labels. If the worst-case label size of $\scheme$ is $L$, then we can create at most $(2^{L+1}-1)^k$ distinct $k$-tuples of labels, so we must have $(m/2)^{k-1}\floor{m/2k}^k \leq (2^{L+1}-1)^k$. From this it follow that

	\begin{align*}
		L & \geq \floor{\frac{k-1}{k}(\log m-1) + \log \floor{m/2k}} \\
		& \geq \floor{\frac{(k-1)^2}{k}+ k -\log k} - 2 \\
		& \geq 2k - \floor{\log k} -4\\
		&= 2\floor{\log n} - \floor{\log\floor{\log n}} -4. \qedhere
	\end{align*}
\end{proof}


\section{Exact distances in trees}

\subsection{Upper bound}
Let $u,v$ be nodes in a tree $T$ and let $w$ be their nearest common ancestor. We then have

\begin{equation}\label{computdist}
	\dist(u,v) = \distroot(u) - \distroot(v) + 2\dist(w,v)
\end{equation} 
If $w=u$ so that $u$ is an ancestor of $v$, then the above equation is just a difference of distroots, which can be stored for each node with $\log n$ bits. The same observation clearly holds if $w=v$.

Assume now that $w\notin\{u,v\}$ so that $u$ and $v$ are not ancestors of each other. 
Consider the heavy-light decomposition~\cite{sleatortarjan} described in the preliniaries. At least one of the nodes $u$ and $v$ must have an ancestor which is a light child of $w$. Assume that it is $v$. Now, $v$ has at most $\log n$ light ancestors. Saving the distance to all of them together with $\distroot$ gives us sufficient information to compute the distance between $u$ and $v$ using equation~\eqref{computdist}. This is the idea behind \Cref{theo:disttrees} below.

By examining the NCA labeling scheme from \cite{AGKR04,alstrupnca2014}, we see that it can easily be extended as follows.
\begin{lemma}[\cite{AGKR04,alstrupnca2014}] \label{lemma:distfromlights}
There exists an NCA labeling scheme of size $O(\log n)$. For any two nodes $u,v$ the scheme returns the label of $w=\nca(u,v)$ as well as:
\begin{itemize}
	\item which of $u$ and $v$ (if any) have a light ancestor that is a child of $w$; and 
	\item the number of light nodes on the path from the root to $w$ and from $w$ to $u$ and $v$, respectively.
\end{itemize} 
\end{lemma}

\begin{theorem}\label{theo:disttrees}
There exists a distance labeling scheme for trees with worst-case label size $\frac{1}{2}\log^2 n+O(\log n)$.
\end{theorem}
\begin{proof}
We use $\Oh(\log n)$ bits for the extended NCA labeling in  \Cref{lemma:distfromlights} and for $\distroot$. 
Using~\eqref{computdist} it now only remains to efficiently represent, for each node, the distance to all its light ancestors.
We consider the light ancestors of a node $v$ encountered on the path from the root to $v$. The distance from $v$ to the root 
is at most $n-1$ and can therefore be encoded with \emph{exactly} $\ceil{\log n}$ bits (by adding leading zeros if needed). By construction of the heavy-light decomposition, the next light node on the path to $v$ will be the root of a subtree of size at most $n/2$, meaning that the distance from $v$ to that ancestor is at most $n/2-1$ and can be encoded with \emph{exactly} $\ceil{\log n}-1$ bits. Continuing this way, we encode the $i$'th light ancestor on the path from the root to $v$ with exactly $\ceil{\log n}-i$ bits. When we run out of light ancestors, we concatenate all the encoded distances, resulting in a string of length at most
\[
\ceil{\log n} + (\ceil{\log n}-1) + \cdots + 2 + 1 = \frac{1}{2}\ceil{\log n }^2 +\frac{1}{2}\ceil{\log n }.
\]
We can use $O(\log n)$ extra bits to encode $n$ and to separate all sublabels from each other. The decoder can now determine $\ceil{\log n}$ and split up the entries in the list of distances. When applying formula~\eqref{computdist}, it can then determine the distance between $v$ and $w$ by adding together the relevant distances in the list of light ancestors, using the fact from \Cref{lemma:distfromlights} that it knows the number of light ancestors from $v$ to $w$.
\end{proof}


\subsection{Lower bound}

In the case of general trees, Gavoille et al~\cite{Gavoille200485} establish a lower bound  of $\frac{1}{8}\log^2 n-O(\log n)$ using an ingenious technique where they apply a distance labeling scheme to a special class of trees called $(h,M)$-trees\footnote{Note that their exposition has some minor errors as pointed out (and corrected) in~\cite{esbenthesis}}. The following uses a generalization of $(h,M)$-trees to improve their ideas and leads to a lower bound of $\frac{1}{4}\log^2 n-O(\log n)$.

\paragraph{$(h,W,a)$-trees.}
We begin with some definitions. For integers $h,W\geq 0$ and a number $a\geq 1$ such that $W/a^i$ is integral for all $i=0,\dots ,h$, an \emph{$(h,W,a)$-tree} is a rooted binary tree $T$ with edge weights in $[0,W]$ that is constructed recursively as follows. For $h=0$, $T$ is just a single node. For $h=1$, $T$ is a claw (i.e.\ a star with three edges) with edge weights $x,x,W-x$ for some $0\leq x<W$ rooted at the leaf node of the edge with weight $W-x$. For $h>1$, $T$ consists of an $(1,W,a)$-tree whose two leaves are the roots of two $(h-1,W/a,a)$-trees $T_0,T_1$.  We shall denote an $(h,W,a)$-tree constructed in this way by $T=\langle T_0,T_1,x\rangle$
An example for $h=3$ can be seen in \Cref{fig:hMtree}. Note that the case $a=1$ simply corresponds to the $(h,W)$-trees defined in~\cite{Gavoille200485}.

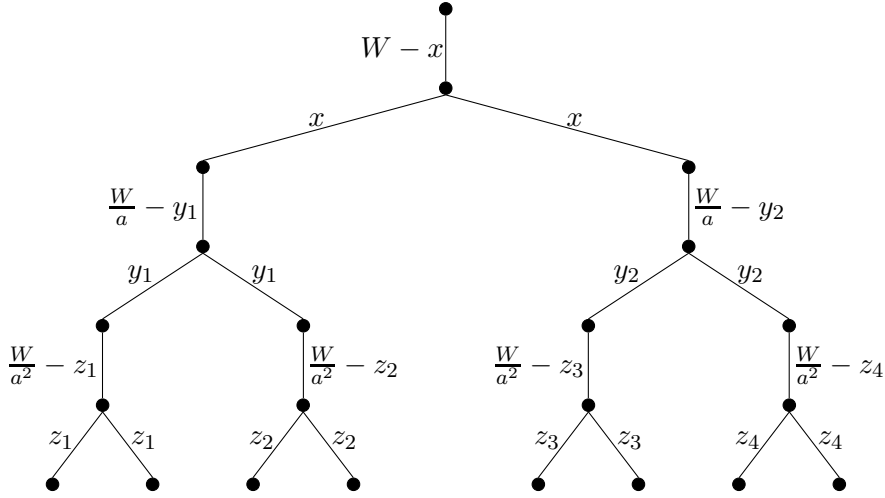
\begin{figure}
	\centering
	\begin{tikzpicture}[sibling distance=3em, inner sep=0pt, minimum size=1ex]
	\tikzstyle{every node}=[draw,circle,fill=black]
	\Tree [.{} \edge node[draw=none,fill=none,auto=right] {$W-x$}; [.{} 
		\edge node[draw=none,fill=none,auto=right] {$x$}; [.{} \edge node[draw=none,fill=none,auto=right] {$\frac{W}{a}-y_1$}; [.{} 
			\edge node[draw=none,fill=none,auto=right] {$y_1$}; [.{} \edge node[draw=none,fill=none,draw=none,fill=none,auto=right] {$\frac{W}{a^2}-z_1$}; [.{} \edge node[draw=none,fill=none,auto=right,pos=.6] {$z_1$}; [ .{} ] \edge node[draw=none,fill=none,auto=left,pos=.6] {$z_1$}; [ .{} ] ] ]
			\edge node[draw=none,fill=none,auto=left] {$y_1$}; [.{} \edge node[draw=none,fill=none,auto=left] {$\frac{W}{a^2}-z_2$}; [.{} \edge node[draw=none,fill=none,auto=right,pos=.6] {$z_2$}; [ .{} ] \edge node[draw=none,fill=none,auto=left,pos=.6] {$z_2$}; [ .{} ] ] ]
			]
		]
		\edge node[draw=none,fill=none,auto=left] {$x$}; [.{}	\edge node[draw=none,fill=none,auto=left] {$\frac{W}{a}-y_2$}; [.{}
			\edge node[draw=none,fill=none,auto=right] {$y_2$}; [.{} \edge node[draw=none,fill=none,auto=right] {$\frac{W}{a^2}-z_3$}; [.{} \edge node[draw=none,fill=none,auto=right,pos=.6] {$z_3$}; [ .{} ] \edge node[draw=none,fill=none,auto=left,pos=.6] {$z_3$}; [ .{} ] ] ]
			\edge node[draw=none,fill=none,auto=left] {$y_2$}; [.{} \edge node[draw=none,fill=none,auto=left] {$\frac{W}{a^2}-z_4$}; [.{} \edge node[draw=none,fill=none,auto=right,pos=.6] {$z_4$}; [ .{} ] \edge node[draw=none,fill=none,auto=left,pos=.6] {$z_4$}; [ .{} ] ] ] 
			] 
		] 
	] ]
	\end{tikzpicture}
	\caption{An $(h,W,a)$-tree, where $h=3$. We require that $x<W$, $y_1,y_2<W/a$ and $z_1,\dots ,z_4<W/a^2$.
}
	\label{fig:hMtree}
\end{figure}

It is easy to see that an $(h,W,a)$-tree
has $2^h$ leaves and $3\cdot 2^h-2$ nodes.  Further, it is straightforward to see that, if $u,v$ are leaves in an $(h,W,a)$-tree $T=\langle T_0,T_1,x\rangle$, then
\begin{equation} \label{eq:disthM}
\dist_T(u,v)= \begin{cases} 2W\frac{a^{-1}-a^{-h}}{1-a^{-1}} + 2x,& \text{if $u\in T_0$ and $v\in T_1$, or vice versa,} \\ \dist_{T_i}(u,v), & \text{if $u,v\in T_i$ for some $i=0,1$}.\end{cases}
\end{equation}

\paragraph{Leaf distance labeling schemes.}
In the following we shall consider \emph{leaf distance labeling schemes} for the family of $(h,W,a)$-trees: that is, distance labeling schemes where only the leaves in a tree need to be labeled, and where only leaf labels can be given as input to the decoder. Since an ordinary distance labeling scheme obviously can be used only for leaves, any lower bound on worst-case label sizes for a leaf distance labeling scheme is also a lower bound for an ordinary distance labeling scheme. We denote by $g(h,W,a)$ the smallest number of labels needed by an optimal leaf distance labeling scheme to label all $(h,W,a)$-trees.
\begin{lemma} \label{lemm:distancehM}
For all $h\geq 1$ and $W\geq 2$, $g(h,W,a)^2\geq Wg(h-1,W^2/a^2,a^2)$.
\end{lemma}
\begin{proof}
Fix an optimal leaf distance labeling scheme $\scheme$ which produces exactly $g(h,W,a)$ distinct labels for the family of $(h,W,a)$-trees. For leaves $u$ and $v$ in an $(h,W,a)$-tree, denote by $l(u)$ and $l(v)$, respectively, the labels assigned by $\scheme$. For $x=0,\dots ,W-1$, let $S(x)$ be the set consisting of pairs of labels $(l(u),l(v))$ for all leaves $u\in T_0$ and $v\in T_1$ in all $(h,W,a)$-trees $T=\langle T_0,T_1,x\rangle$.

The sets $S(x)$ and $S(x')$ are disjoint for $x\neq x'$, since every pair of labels in $S(x)$ uniquely determines $x$ due to~\eqref{eq:disthM}. Letting $S=\bigcup_{x=0}^{W-1}S(x)$, we therefore have $|S|=\sum_{x=0}^{W-1}|S(x)|$. 
Since $S$ contains pairs of labels produced by $\scheme$ from leaves in $(h,W,a)$-trees , we clearly also have $|S|\leq g(h,W,a)^2$, and hence it only remains to prove that $|S|\geq Wg(h-1,W^2/a^2,a^2)$, which we shall do by showing that $|S(x)|\geq g(h-1,W^2/a^2,a^2)$ for all $x$.

The goal for the rest of the proof is therefore to create a leaf distance labeling scheme for $(h-1,W^2/a^2,a^2)$-trees using only labels from the set $S(x)$ for some fixed $x$. So let $x$ be given and consider an $(h-1,W^2/a^2,a^2)$-tree $T'$. Let  $V=W/a$. From $T'$ we shall construct an $(h-1,V,a)$-tree $\phi_i(T')$ for $i=0,1$ such that every leaf node $v$ in $T'$ corresponds to nodes $\phi_i(v)$ in $\phi_i(T')$ for $i=0,1$.
The trees $\phi_i(T')$ are defined as follows.
If $h=1$, so that $T'$ consists of a single node, then $\phi_i(T')=T'$ for $i=0,1$. 
If $h>1$, then $T'$ is in the form $T'=\langle T'_0,T'_1,y\rangle$ for some $0\leq y< V^2$. We can write $y$ in the form $y=y_0+y_1V$ for uniquely determined $y_0,y_1$ with $0\leq y_0,y_1<V$. For $i=0,1$, we recursively define $\phi_i(T') = \langle \phi_i(T'_0), \phi_i(T'_1),y_i\rangle$. Thus, $\phi_i(T')$ is an $(h-1,V,a)$-tree that is similar to $T'$ but where we replace the top edge weight $y$ by edge weights $y_i$ and, recursively, do the same for all $(h-2,V^2/a^2,a^2)$-subtrees. Note also that the corresponding edge weight $V^2-y$ in $T'$ automatically is replaced by the edge weight $V-y_i$ in $\phi_i(T')$ in order for $\phi_i(T')$ to be an $(h-1,V,a)$-tree.

Denote by $\phi_i(v)$ the leaf in $\phi_i(T')$ corresponding to the leaf $v$ in $T'$.

Consider now the $(h,W,a)$-tree $T=\langle \phi_0(T'),\phi_1(T'),x\rangle$. Every leaf $v$ in $T'$ corresponds to the leaves $\phi_0(v),\phi_1(v)$ in $T$ where $\phi_i(v)\in \phi_i(T')$ for $i=0,1$. 
Using  formula~\eqref{eq:disthM} for the distances in $T'$, it is straightforward to see that
\begin{equation*} \label{eq:disthMdistances}
\dist_{T'}(u,v) = \left(\dist_{\phi_0(T')}(\phi_0(u),\phi_0(v)) \bmod (2V)\right) + V\dist_{\phi_1(T')}(\phi_1(u),\phi_1(v)).
\end{equation*}

We can now apply the leaf distance labeling scheme $\scheme$ to $T$ and obtain a label for each leaf node in $T$. In particular, the pair of leaves $(\phi_0(v),\phi_1(v))$ corresponding to a node $v$ in $T'$ will receive a pair of labels. We use this pair to label $v$ in $T'$, whereby we have obtained a labeling of the leaves in $T'$ with labels from $S(x)$. Using the formula in~\eqref{eq:disthMdistances} we can construct a decoder that can compute the distance between two nodes in $T'$ using these labels alone, and hence we have obtained a leaf distance labeling scheme for $(h-1,V^2,a^2)$-trees using only labels from $S(x)$ as desired.
\end{proof}

\begin{lemma} \label{lemm:distancehM2}
For all $h\geq 1$ and $W\geq 2$, $g(h,W,a)\geq \frac{W^{h/2}}{a^{h(h-1)/4}}$.
\end{lemma}
\begin{proof}
The proof is by induction on $h$. For $h=1$ we note that an $(0,W,a)$-tree has only one node, so that $g(0,W^2/a^2,a^2)=1$. \Cref{lemm:distancehM} therefore yields $g(1,W,a)^2\geq W$ from which it follows that $g(1,W,a)\geq \sqrt{W}$. The lemma therefore holds for $h=1$. Now let $h>1$ and assume that the lemma holds for $h-1$. \Cref{lemm:distancehM} and the induction hypothesis now yield
\begin{align*}
g(h,W,a)^2 &\geq Wg(h-1,W^2/a^2,a^2) \\
 & \geq W  \frac{(W^2/a^2)^{(h-1)/2}}{a^{2(h-1)(h-2)/4}}  \\
 &= \frac{W^h}{a^{h(h-1)/2}}
\end{align*}
from which it follows that $g(h,W,a)\geq  \frac{W^{h/2}}{a^{h(h-1)/4}}$ as desired.
\end{proof}

The previous lemma implies that any (leaf and hence also ordinary) distance labeling scheme for $(h,W,a)$-trees must have labels with worst-case length at least $\frac{h}{2}(\log W-\frac{h-1}{2}\log a) = \frac{1}{2}h\log W-\frac{1}{4}h^2\log a+\frac{1}{4}h\log a$. Since the number of nodes in such a tree is $n=3\cdot 2^h-2$, it follows that $h=\log (n+2)-\log 3$, and hence that $\log n-2\leq h\leq \log n$ for sufficiently large $n$. From this we see that the worst case label length is at least 
\[
\frac{1}{2}\log n\log W-\frac{1}{4}\log n (\log n - 1)\log a-\log W -\frac{1}{2}\log a.
\]
In the case where $a=1$, we retrieve the bound of $\frac{1}{2}\log n\log W-\log W$ obtained in~\cite{gavoillepelegperennesraz}. It seems that larger values of $a$ only makes the above result weaker, but the the real strength of the above becomes apparent when we switch to the unweighted version of $(h,W,a)$-trees, in which we replace weighted edges by paths of similar lenghts. Note that a distance labeling scheme for the family of unweighted $(h,W,a)$-trees can be used as a distance labeling scheme for the weighted $(h,W,a)$-trees, and hence any lower bound in the weighted version automatcially becomes a lower bound in the unweighted version.

The number of nodes $n$ in an \emph{unweighted} $(h,W,a)$-tree is upper bounded by
\[
n\leq  2W+2\cdot 2W/a + 2^2\cdot 2W/a^2 + \dots + 2^{h-1}\cdot 2W/a^{h-1} +1
\]
In the case $a=2$, we get $n \leq 2Wh+1$.

\begin{theorem} \label{theo:distancelowerbintrees}
Any distance labeling scheme for unweighted $(h,W,2)$-trees, and hence also for general trees, has a worst-case label size of at least $\frac{1}{4}\log^2 n-O(\log n)$.
\end{theorem}
\begin{proof}
Choose the largest integer $h$ with $2\cdot 2^hh+1\leq n$, and note that we must have $h\geq \log n-O(\log\log n)$.  Set $W=2^h$ and consider the family of  $(h,W,2)$-trees, which is a subfamily of the family of trees with $n$ nodes. From \Cref{lemm:distancehM2} it therefore follows that the worst-case label length is
\begin{align*}
\frac{1}{2}h\log W-\frac{1}{4}h^2+\frac{1}{4}h &= \frac{1}{4}h^2 +\frac{1}{4}h 
 = \frac{1}{4}\log^2 n + \frac{1}{4}\log n -O(\log\log  n). \qedhere
\end{align*}
\end{proof}


\section{Approximate distances in trees}
In this section we present a $(1+\eps)$-stretch distance labeling schemes with labels of size $\Oh(\log n)$.  
\begin{theorem}
   For constant $\eps>0$, $1+\eps$ stretch labeling scheme use $\Theta(\log n)$ bits.
\end{theorem}  
\begin{proof}
As in the case of exact distances, we will create labels of size $\Oh(\log n)$ bits that contain the extended NCA labels from \Cref{lemma:distfromlights} as well as $\distroot$. We will also be using the formula in~\eqref{computdist}. However we can not afford to store exact distance to each apex ancestor. Even storing an $2$-approximate distance to each apex ancestor would require $\log n \log \log n$ bits. Furthermore, given approximate distance to the apex nodes does not directly guarantee upper bound for the approximate distance, as we in equation~(\ref{computdist}) are using subtractions. We will in the following address these two problems.

Let $w=\nca(u,v)$ and assume $w \not \in \{u,v\}$, since otherwise we can compute the exact distance using only $\distroot$.
Suppose we know a $(1+\eps)$-approximation $\alpha$ of $\dist(w,v)$ for some $\eps \geq 0$. That is, 
\begin{equation}\label{bound-alpha}
	\dist(w,v) \leq \alpha \leq (1+\eps) \dist(w,v).
\end{equation} 
Define $\tilde{d}=\distroot(u) - \distroot(v) + 2\alpha$. First we show that $\tilde d$ is a $(1+2\eps)$-approximation of $\dist(u,v)$. Next we show how to represent all the $(1+\eps)$-approximate distances to light ancestors for a node using a total of $\Oh(\log n)$ bits.
Together with formula~\eqref{computdist}, these two facts prove that we can compute $(1+2\eps)$-stretch distances between any pair of nodes with labels of size $O(\log n)$. To prove the theorem, we can then simple replace $\eps$ by $\frac{1}{2}\eps$.

To see that  $\tilde d$ is a $(1+2\eps)$-approximation of $\dist(u,v)$, first note that 
\[
\tilde{d} = \distroot(u) - \distroot(v) + 2\alpha \geq \distroot(u) - \distroot(v) + 2\dist(w,v) = \dist(u,v).
\]
 For the other inequality, note that
\begin{eqnarray*}
	\tilde{d} &=& \distroot(u) - \distroot(v) + 2\alpha \\
	&\leq& \distroot(u) - \distroot(v) + 2(1+\epsilon)\dist(w,v) \\
	&=  & \distroot(u) - (\distroot(v) - \dist(w,v)) + (1+2\epsilon) \dist(w,v) \\
	&=& \distroot(u) - \distroot(w) + (1+2\epsilon)\dist(w,v)\\
	&=& \dist(u,w) + (1+2\epsilon)\dist(w,v) \\
	&\leq& (1+2\epsilon)\left( \dist(u,w) + \dist(w,v) \right) \\
	&=& (1+2\epsilon)\dist(u,v)\;.
\end{eqnarray*} 

It now only remains to show that we can compactly store all the approximate distances $\alpha$ to light ancestors using $O(\log n)$ bits space.

We use a heavy light path decomposition of the tree. For each node $v$ we can save a $2$ approximate distance to all its $k$ proper light ancestors node as follows. Let $S$ be a binary string initially with $k$ zeros. Before each $0$ we now inserts $1$s such that, if we have $j$ $1$s in total from the beginning of $S$ and to the $i$'th $0$, then the distance to the $i$th light ancestor $a$ of $v$ satisfies that $2^{j-1}\leq\dist(v,a)\leq 2^j$. This is the same as traversing the tree bottom-up from $v$ and, for each light node encountered on the way, adding a $0$ and each time the distance doubles adding a $1$. The number of $0$s equal the number of light nodes which is at most $\log n$, and the number of $1$s is also limited by $\log n$ since $n$ is the maximum distance in the tree. In total the length of $S$ is at most $2\log n$. 

Using the $O(\log n)$ bits label from Lemma~\ref{lemma:distfromlights} we can tell if one node is an ancestor of another, and if not which one has a light ancestor $a$ that is a child of their nearest common ancestor $w$. In addition, we can determine the total number $i$ of light ancestors up to $a$. This means that we can compute $j$, and hence the 2-approximation $j-1$, as the number of $1$'s in $S$ until the $i$'th $0$. 

We have now obtained a 2-approximation with labels of size $O(\log n)$. We can improve this to a $(1+\eps)$-approximation by setting a $1$ in $S$ each time the distance increases with $1+\eps$ rather than $2$. This will increase the label size with a constant factor $\frac{1}{\log (1+\eps)}$.

This proves that there is a $(1+\eps)$-stretch distance labeling scheme with $O(\log n)$. To complete the proof of the theorem, we note that, given any $(1+\eps)$-stretch distance scheme, we can always distinguish nodes (since identical nodes have distance 0), which means that we always need at least $n$ different labels, and hence labels of size at least $\log n$ bits.
\end{proof}

\newpage

\bibliographystyle{plain}
\bibliography{labels}

\begin{thebibliography}{10}

\bibitem{abiteboul}
S.~Abiteboul, S.~Alstrup, H.~Kaplan, T.~Milo, and T.~Rauhe.
\newblock Compact labeling scheme for ancestor queries.
\newblock {\em SIAM J. Comput.}, 35(6):1295--1309, 2006.

\bibitem{AKM01}
S.~Abiteboul, H.~Kaplan, and T.~Milo.
\newblock Compact labeling schemes for ancestor queries.
\newblock In {\em Proc. of the 12th Annual ACM-SIAM Symp. on Discrete
  Algorithms (SODA)}, pages 547--556, 2001.

\bibitem{AlstrupBR03}
S.~Alstrup, P.~Bille, and T.~Rauhe.
\newblock Labeling schemes for small distances in trees.
\newblock In {\em Proc. of the 14th Annual ACM-SIAM Symp. on Discrete
  Algorithms (SODA)}, pages 689--698, 2003.

\bibitem{alstrupbillerauhe}
S.~Alstrup, P.~Bille, and T.~Rauhe.
\newblock Labeling schemes for small distances in trees.
\newblock {\em SIAM J. Discrete Math.}, 19(2):448--462, 2005.
\newblock See also SODA'03.

\bibitem{adjacencytrees2015}
S.~Alstrup, S.~Dahlgaard, and M.~B.~T. Knudsen.
\newblock Optimal induced universal graphs and labeling schemes for trees.
\newblock In {\em Proc. 56th Annual Symp. on Foundations of Computer Science
  (FOCS)}, 2015.

\bibitem{distalstrup}
S.~Alstrup, C.~Gavoile, E.~B. Halvorsen, and H.~Petersen.
\newblock Simpler, faster and shorter labels for distances in graphs.
\newblock {\em Submitted}, 2015.

\bibitem{AGKR04}
S.~Alstrup, C.~Gavoille, H.~Kaplan, and T.~Rauhe.
\newblock Nearest common ancestors: A survey and a new algorithm for a
  distributed environment.
\newblock {\em Theory of Computing Systems}, 37(3):441--456, May 2004.

\bibitem{alstrupnca2014}
S.~Alstrup, E.~B. Halvorsen, and K.~G. Larsen.
\newblock Near-optimal labeling schemes for nearest common ancestors.
\newblock In {\em Proc. of the 25th Annual ACM-SIAM Symp. on Discrete
  Algorithms (SODA)}, pages 972--982, 2014.

\bibitem{AlstrupKTZ14}
S.~Alstrup, H.~Kaplan, M.~Thorup, and U.~Zwick.
\newblock Adjacency labeling schemes and induced-universal graphs.
\newblock In {\em Proc. of the 47th Annual ACM Symp. on Theory of Computing
  (STOC)}, 2015.

\bibitem{AR02}
S.~Alstrup and T.~Rauhe.
\newblock Improved labeling schemes for ancestor queries.
\newblock In {\em Proc. of the 13th Annual ACM-SIAM Symp. on Discrete
  Algorithms (SODA)}, 2002.

\bibitem{alstruprauhe}
S.~Alstrup and T.~Rauhe.
\newblock Small induced-universal graphs and compact implicit graph
  representations.
\newblock In {\em Proc. 43rd Annual Symp. on Foundations of Computer Science
  (FOCS)}, pages 53--62, 2002.

\bibitem{BG09}
F.~Bazzaro and C.~Gavoille.
\newblock Localized and compact data-structure for comparability graphs.
\newblock {\em Discrete Mathematics}, 309(11):3465--3484, 2009.

\bibitem{bonichon2006short}
N.~Bonichon, C.~Gavoille, and A.~Labourel.
\newblock Short labels by traversal and jumping.
\newblock In {\em Structural Information and Communication Complexity}, pages
  143--156. Springer, 2006.
\newblock Include proof for binary trees and caterpillars.

\bibitem{Breuer66}
M.~A. Breuer.
\newblock Coding the vertexes of a graph.
\newblock {\em IEEE Trans. on Information Theory}, IT--12:148--153, 1966.

\bibitem{BF67}
M.~A. Breuer and J.~Folkman.
\newblock An unexpected result on coding vertices of a graph.
\newblock {\em J. of Mathemathical analysis and applications}, 20:583--600,
  1967.

\bibitem{CDEHV08}
V.~D. Chepoi, F.~F. Dragan, B.~Estellon, M.~Habib, and Y.~Vax{\`e}s.
\newblock Diameters, centers, and approximating trees of delta-hyperbolic
  geodesic spaces and graphs.
\newblock In {\em 24st Annual ACM Symp. on Computational Geometry (SoCG)},
  pages 59--68, 2008.

\bibitem{CDV06}
V.~D. Chepoi, F.~F. Dragan, and Y.~Vax\`es.
\newblock Distance and routing labeling schemes for non-positively curved plane
  graphs.
\newblock {\em J. of Algorithms}, 61(2):60--88, 2006.

\bibitem{Chung90}
F.~R.~K. Chung.
\newblock Universal graphs and induced-universal graphs.
\newblock {\em J. of Graph Theory}, 14(4):443--454, 1990.

\bibitem{CohenKaplan2010}
E.~Cohen, H.~Kaplan, and T.~Milo.
\newblock Labeling dynamic {XML} trees.
\newblock {\em SIAM J. Comput.}, 39(5):2048--2074, 2010.

\bibitem{CV03}
B.~Courcelle and R.~Vanicat.
\newblock Query efficient implementation of graphs of bounded clique-width.
\newblock {\em Discrete Applied Mathematics}, 131:129--150, 2003.

\bibitem{Cowen01}
L.~J. Cowen.
\newblock Compact routing with minimum stretch.
\newblock {\em J. of Algorithms}, 38:170--183, 2001.
\newblock See also SODA'91.

\bibitem{DPT10}
Y.~Dodis, M.~P\v{a}tra\c{s}cu, and M.~Thorup.
\newblock Changing base without losing space.
\newblock In {\em Proc. of the 42nd Annual ACM Symp. on Theory of Computing
  (STOC)}, pages 593--602, 2010.

\bibitem{EilamGP03}
T.~Eilam, C.~Gavoille, and D.~Peleg.
\newblock Compact routing schemes with low stretch factor.
\newblock {\em J. of Algorithms}, 46(2):97--114, 2003.

\bibitem{ElkinFN15}
M.~Elkin, A.~Filtser, and O.~Neiman.
\newblock Prioritized metric structures and embedding.
\newblock In {\em Proc. of the 47th Annual ACM Symp. on Theory of Computing
  (STOC)}, pages 489--498, 2015.

\bibitem{FarzanM13}
A.~Farzan and J.~I. Munro.
\newblock Succinct encoding of arbitrary graphs.
\newblock {\em Theoretical Computer Science}, 513:38--52, 2013.

\bibitem{FarzanM14}
A.~Farzan and J.~I. Munro.
\newblock A uniform paradigm to succinctly encode various families of trees.
\newblock {\em Algorithmica}, 68(1):16--40, 2014.

\bibitem{Fraigniaud2009randomized}
P.~Fraigniaud and A.~Korman.
\newblock On randomized representations of graphs using short labels.
\newblock In {\em Proc. of the 21st Annual Symp. on Parallelism in Algorithms
  and Architectures (SPAA)}, pages 131--137, 2009.

\bibitem{fraigniaudkorman2}
P.~Fraigniaud and A.~Korman.
\newblock Compact ancestry labeling schemes for {XML} trees.
\newblock In {\em Proc. of the 21st annual ACM-SIAM Symp. on Discrete
  Algorithms (SODA)}, pages 458--466, 2010.

\bibitem{fraigniaudkorman}
P.~Fraigniaud and A.~Korman.
\newblock An optimal ancestry scheme and small universal posets.
\newblock In {\em Proc. of the 42nd Annual ACM Symp. on Theory of Computing
  (STOC)}, pages 611--620, 2010.

\bibitem{GKKPP01}
C.~Gavoille, M.~Katz, N.~Katz, C.~Paul, and D.~Peleg.
\newblock Approximate distance labeling schemes.
\newblock In {\em Proc. of the 9th Annual European Symp. on Algorithms (ESA)},
  pages 476--488, 2001.

\bibitem{GavoilleL07IS}
C.~Gavoille and A.~Labourel.
\newblock Distributed relationship schemes for trees.
\newblock In {\em 18th International Symp. on Algorithms and Computation
  (ISAAC)}, pages 728--738, 2007.

\bibitem{GavoilleL07PO}
C.~Gavoille and A.~Labourel.
\newblock On local representation of distances in trees.
\newblock In {\em Proc. of the 26th Annual ACM Symp. on Principles of
  Distributed Computing (PODC)}, pages 352--353, 2007.

\bibitem{GL05}
C.~Gavoille and O.~Ly.
\newblock Distance labeling in hyperbolic graphs.
\newblock In {\em 16th Annual International Symp. on Algorithms and Computation
  (ISAAC)}, pages 1071--1079, 2005.

\bibitem{GP03b}
C.~Gavoille and C.~Paul.
\newblock Distance labeling scheme and split decomposition.
\newblock {\em Discrete Mathematics}, 273(1-3):115--130, 2003.

\bibitem{GP08}
C.~Gavoille and C.~Paul.
\newblock Optimal distance labeling for interval graphs and related graphs
  families.
\newblock {\em SIAM J. Discrete Math.}, 22(3):1239--1258, 2008.

\bibitem{gavoillepelegperennesraz}
C.~Gavoille, D.~Peleg, S.~P{\'e}rennes, and R.~Raz.
\newblock Distance labeling in graphs.
\newblock In {\em Proc. of the 12th Annual ACM-SIAM Symp. on Discrete
  algorithms (SODA)}, pages 210--219, 2001.

\bibitem{Gavoille200485}
C.~Gavoille, D.~Peleg, S.~P{\'e}rennes, and R.~Raz.
\newblock Distance labeling in graphs.
\newblock {\em J. of Algorithms}, 53(1):85 -- 112, 2004.
\newblock See also SODA'01.

\bibitem{grahampollak}
R.~L. Graham and H.~O. Pollak.
\newblock On embedding graphs in squashed cubes.
\newblock In {\em Lecture Notes in Mathematics}, volume 303. Springer-Verlag,
  1972.

\bibitem{GuptaKL03}
A.~Gupta, R.~Krauthgamer, and J.~R. Lee.
\newblock Bounded geometries, fractals, and low-distortion embeddings.
\newblock In {\em 44th Annual Symp. on Foundations of Computer Science (FOCS)},
  pages 534--543, 2003.

\bibitem{gupta2005traveling}
A.~Gupta, A.~Kumar, and R.~Rastogi.
\newblock Traveling with a pez dispenser (or, routing issues in mpls).
\newblock {\em SIAM J. on Computing}, 34(2):453--474, 2005.
\newblock See also FOCS'01.

\bibitem{esbenthesis}
E.~B. Halvorsen.
\newblock Labeling schemes for trees - overview and new results.
\newblock Master's thesis, University of Copenhagen, 2013.
\newblock Available at {\tt esben.bistruphalvorsen.dk}.

\bibitem{KNR92}
S.~Kannan, M.~Naor, and S.~Rudich.
\newblock Implicit representation of graphs.
\newblock {\em SIAM J. Disc. Math.}, pages 596--603, 1992.
\newblock See also STOC'88.

\bibitem{KaoLW07}
M.~Kao, X.~Li, and W.~Wang.
\newblock Average case analysis for tree labelling schemes.
\newblock {\em Theor. Comput. Sci.}, 378(3):271--291, 2007.

\bibitem{KM01}
H.~Kaplan and T.~Milo.
\newblock Short and simple labels for distances and other functions.
\newblock In {\em 7nd Work. on Algo. and Data Struc.}, 2001.

\bibitem{KMS02}
H.~Kaplan, T.~Milo, and R.~Shabo.
\newblock A comparison of labeling schemes for ancestor queries.
\newblock In {\em Proc. of the 13th Annual ACM-SIAM Symp. on Discrete
  Algorithms (SODA)}, 2002.

\bibitem{siamcompKatzKKP04}
M.~Katz, N.~A. Katz, A.~Korman, and D.~Peleg.
\newblock Labeling schemes for flow and connectivity.
\newblock {\em SIAM J. Comput.}, 34(1):23--40, 2004.
\newblock See also SODA'02.

\bibitem{Korman2010}
A.~Korman.
\newblock Labeling schemes for vertex connectivity.
\newblock {\em ACM Trans. Algorithms}, 6(2):39:1--39:10, 2010.

\bibitem{KormanK13}
A.~Korman and S.~Kutten.
\newblock Labeling schemes with queries.
\newblock {\em CoRR}, abs/cs/0609163, 2006.

\bibitem{KormanK07}
A.~Korman and S.~Kutten.
\newblock Labeling schemes with queries.
\newblock In {\em SIROCCO}, pages 109--123, 2007.

\bibitem{dynamicKormanP07}
A.~Korman and D.~Peleg.
\newblock Labeling schemes for weighted dynamic trees.
\newblock {\em Inf. Comput.}, 205(12):1721--1740, 2007.

\bibitem{KL06}
R.~Krauthgamer and J.~R. Lee.
\newblock Algorithms on negatively curved spaces.
\newblock In {\em 47th Annual Symp. on Foundations of Computer Science (FOCS)},
  pages 119--132, 2006.

\bibitem{moon1965minimal}
J.~W. Moon.
\newblock On minimal $n$-universal graphs.
\newblock {\em Proc. of the Glasgow Mathematical Association}, 7(1):32--33,
  1965.

\bibitem{muller}
J.~H. M\"{u}ller.
\newblock {\em Local structure in graph classes}.
\newblock PhD thesis, Georgia Institute of Technology, 1988.

\bibitem{MunroRRR12}
J.~I. Munro, R.~Raman, V.~Raman, and S.~Srinivasa Rao.
\newblock Succinct representations of permutations and functions.
\newblock {\em Theor. Comput. Sci.}, 438:74--88, 2012.

\bibitem{patrascu08succinct}
M.~P{\v a}tra{\c s}cu.
\newblock Succincter.
\newblock In {\em Proc. 49th Annual Symp. on Foundations of Computer Science
  (FOCS)}, pages 305--313, 2008.

\bibitem{peleg2}
D.~Peleg.
\newblock Informative labeling schemes for graphs.
\newblock In {\em Proc. 25th Symp. on Mathematical Foundations of Computer
  Science}, pages 579--588, 2000.

\bibitem{peleg}
D.~Peleg.
\newblock Proximity-preserving labeling schemes.
\newblock {\em J. Graph Theory}, 33(3):167--176, 2000.

\bibitem{SK85}
N.~Santoro and R.~Khatib.
\newblock Labeling and implicit routing in networks.
\newblock {\em The computer J.}, 28:5--8, 1985.

\bibitem{sleatortarjan}
D.~D. Sleator and R.~E. Tarjan.
\newblock A data structure for dynamic trees.
\newblock {\em J. of Computer and System Sciences}, 26(3):362 -- 391, 1983.

\bibitem{spinrad2003efficient}
J.~P. Spinrad.
\newblock {\em Efficient Graph Representations}, volume~19 of {\em Fields
  Institute Monographs}.
\newblock AMS, 2003.

\bibitem{Talwar04}
K.~Talwar.
\newblock Bypassing the embedding: algorithms for low dimensional metrics.
\newblock In {\em Proc. of the 36th Annual ACM Symp. on Theory of Computing
  (STOC)}, pages 281--290, 2004.

\bibitem{TangYZ09}
M.~Tang, J.~Yang, and G.~Zhang.
\newblock A compact distance labeling scheme for trees of small depths.
\newblock In {\em International Conference on Scalable Computing and
  Communications / Eighth International Conference on Embedded Computing,
  ScalCom-EmbeddedCom}, pages 455--458, 2009.

\bibitem{Thorup2004distance}
M.~Thorup.
\newblock Compact oracles for reachability and approximate distances in planar
  digraphs.
\newblock {\em J. ACM}, 51(6):993--1024, 2004.
\newblock See also FOCS'01.

\bibitem{throupzwick}
M.~Thorup and U.~Zwick.
\newblock Compact routing schemes.
\newblock In {\em Proc. of the 13th Annual ACM Symp. on Parallel Algorithms and
  Architectures}, SPAA '01, pages 1--10, 2001.

\bibitem{ThZw05}
M.~Thorup and U.~Zwick.
\newblock Approximate distance oracles.
\newblock {\em J. of the ACM}, 52(1):1--24, 2005.
\newblock See also STOC'01.

\bibitem{weinmannpeleg}
O.~Weimann and D.~Peleg.
\newblock A note on exact distance labeling.
\newblock {\em Inf. Process. Lett.}, 111(14):671--673, 2011.

\bibitem{wiki:implicit}
Wikipedia.
\newblock Implicit graph --- wikipedia{,} the free encyclopedia, 2013.
\newblock [Online; accessed 15-February-2014].

\bibitem{winkler}
P.~M. Winkler.
\newblock Proof of the squashed cube conjecture.
\newblock {\em Combinatorica}, 3(1):135--139, 1983.

\end{thebibliography}

\makeatletter
\def\runninghead{\hrulefill\quad APPENDIX\quad\hrulefill}
\def\ps@headings{
\def\@oddhead{\footnotesize\rm\hfill\runninghead\hfill}}
\def\@evenhead{\@oddhead}
\def\@oddfoot{\rm\hfill\thepage\hfill}\def\@evenfoot{\@oddfoot}
\makeatother

\newpage
\setlength{\headsep}{15pt} \pagestyle{headings}

\appendix

\end{document}